\documentclass{eptcs}
\providecommand{\event}{WRS 2011}
\usepackage{amssymb}
\usepackage{amsmath}
\usepackage{makeidx}

\newcommand{\ifr}{\mbox{\sf if}}
\newcommand{\tr}{\mbox{\sf true}}
\newcommand{\fa}{\mbox{\sf false}}
\newcommand{\fib}{\mbox{\sf fib}}
\newcommand{\ffact}{\mbox{\sf fact}}
\newcommand{\pl}{{\sf plus}}
\newcommand{\ml}{{\sf mult}}
\newcommand{\az}{{\sf .0}}
\newcommand{\ao}{{\sf .1}}
\newcommand{\fs}{{\sf s}}
\newcommand{\fz}{{\sf 0}}
\newcommand{\fo}{{\sf 1}}
\newcommand{\ff}{{\sf f}}
\newcommand{\fg}{{\sf g}}

\newcommand{\scp}{{\sf succ}}
\newcommand{\plp}{{\sf plus}}
\newcommand{\mlp}{{\sf mult}}

\newenvironment{proof}{Proof.}{\hfill $\Box$

\vspace{3mm}

}
\newtheorem{thm}{Theorem}
\newtheorem{lem}[thm]{Lemma}

\title{Strategy Independent Reduction Lengths in Rewriting and
Binary Arithmetic}
\def\titlerunning{Strategy Independent Reduction Lengths and Binary
Arithmetic}

\author{Hans Zantema
    \institute{University of Technology, Eindhoven, The Netherlands}
    \institute{Radboud University, Nijmegen, The Netherlands}
    \email{h.zantema@tue.nl}
}
\def\authorrunning{Hans Zantema}

\begin{document}
\maketitle

\begin{abstract}
In this paper we give a criterion by which one can conclude that every
reduction of a basic term to normal form has the same length. As a
consequence, the number of steps to reach the normal form is independent
of the chosen strategy. In particular this holds for TRSs computing
addition and
multiplication of natural numbers, both in unary and binary notation.
\end{abstract}

\section{Introduction}
\label{sect:introduction}

For many term rewriting systems (TRSs) the number of rewrite steps to
reach a normal form strongly depends on the chosen strategy. For
instance, in using the standard if-then-else-rules
\[ 
\ifr(\tr,x,y) \to  x \;\;\;\; \ifr(\fa,x,y) \to y  \]
it is a good strategy to first rewrite the leftmost (boolean) argument 
of the symbol
$\ifr$ until this argument is rewritten to $\fa$ or $\tr$ and then
apply the corresponding rewrite rule for $\ifr$. In this way redundant
reductions in the third argument are avoided in case the condition
rewrites to $\tr$, and redundant reductions in the second argument are 
avoided in case the condition rewrites to $\fa$. More general, choosing a
good reduction strategy is essential for doing efficient computation by
rewriting. Roughly speaking, for erasing rules, that is, some 
variable in
the left-hand side does not appear in the right-hand side, it
seems a good strategy to postpone rewriting this possibly
erasing argument, as is the case in the above if-then-else
example. Conversely, in case of duplicating rules, that is, a
variable occurs more often in the right-hand side than in
the left-hand side, it seems a good strategy to first rewrite
the corresponding argument before duplicating it.

Surprisingly however, there are practical examples of TRSs
including both such erasing and duplicating rules, where the
strategy has no influence at all on the number steps required to
reach a normal form. In this paper we investigate criteria
for this phenomenon. As an example, using the results of this
paper, we will show that for multiplying two natural numbers by
the standard TRS
\[ \begin{array}{rclrcl}
\pl(\fz,x) & \to & x &
\ml(\fz,x) & \to & \fz \\
\pl(\fs(x),y) & \to & \fs(\pl(x,y)) &
\ml(\fs(x),y) & \to & \pl(\ml(x,y),y) \end{array} \]
the number of steps to reach the resulting normal form is
independent of the chosen strategy. Note that the third rule of
this TRS is erasing, and the last rule is duplicating.

In this TRS ground terms reduce to ground normal forms of the
shape $\fs^n(\fz)$ for any natural number $n$, representing the
corresponding number $n$ in unary notation. For large numbers $n$ 
this is a quite inefficient representation. Much more efficient
is the binary representation. In this paper we also give a TRS
for doing addition and multiplication in binary representation.
Here for every addition or multiplication of two
positive integer numbers we show by the main theorem of this paper
that every reduction to normal form has the same number of steps.
Moreover, this number 
has the same complexity as the standard binary algorithms:
linear for addition and quadratic for  multiplication. 

Throughout the paper we assume familiarity with
standard notions in rewriting like orthogonality and normal
forms as they are introduced in e.g. \cite{terese}. For instance,
a {\em reduction} is a sequence of rewrite steps, and a {\em
reduct} of a term $t$ is a term $u$ for which there is a reduction
from $t$ to $u$.

In Section \ref{secdiam} we investigate the diamond property and
present our main result in Theorem \ref{thmmain}: a criterion  
by which for every basic term all reductions to normal form have
the same length. Here a term is called basic if only its root is a
defined symbol. We apply this to unary arithmetic. In Section
\ref{secbin} it is shown how our theorem applies to binary
arithmetic. We conclude in Section \ref{secconcl}.

\section{The diamond property and the main result}
\label{secdiam}

We say that a binary relation $\to$ satisfies the {\em diamond
property}, 
if for every three elements
$s,t,u$ satisfying $s \to t$ and $s \to u$ and $t \neq u$, there
exists an element $w$ such that $t \to w$ and $u \to w$.

This notion is slightly different from other variants of the
diamond property, for instance introduced in \cite{terese} where
it is typically used for reflexive relations.
However, as we are interested in the exact number of steps to
reach a normal form, our present version is the most natural in our context.
In earlier texts it was sometimes called WCR$^1$.
It is an instance of the {\em balanced weak Church-Rosser property}
from \cite{T05}.  The following lemma is a direct consequence
of both Lemma 1 from \cite{T05} and the random descent lemma from \cite{O07}; 
to make the paper self-contained here we give a proof too.  

%
%
%
%

\begin{lem}
\label{cordiam}
Let $\to$ be a relation satisfying the diamond property. Then
every element has at most one normal form, and for every element
having a normal form it holds that every reduction
of this element to its normal form has the same number of steps.
\end{lem}
\begin{proof}
This follows from the following claim:
\begin{quote}
Let $t$ be an element having normal form $n$. Then every reduct
of $t$ has normal form $n$, and every $\to$-reduction from $t$ to
$n$ has the same number of steps.
\end{quote}
We prove this by induction on $k$, where $k$ is the number of
steps of a reduction from $t$ to $n$.

If $k=0$ then $t$ is a normal form, and the lemma holds.

If $k > 0$ assume $ t \to t_1 \to^{k-1} n$ and $t \to^p u$; we
will prove that $u \to^q n$ for $q$ satisfying $p+q = k$.
Then indeed the lemma follows by applying this claim both for
$u$ being the arbitrary reduct of $t$ and for $u = n$.

For $p=0$ the claim trivially holds, so assume $p > 0$ and
$t \to t_2 \to^{p-1} u$. In case $t_1 = t_2$ we apply the
induction hypothesis to $t_1 = t_2$ and we are done. In the
remaining case by the diamond property there exists $v$ such
that $t_1 \to v$ and $t_2 \to v$. By applying the induction
hypothesis to $t_1$ we conclude $v \to^{k-2} n$, yielding
$t_2 \to^{k-1} n$. Next we apply the induction hypothesis to $t_2$,
so $u \to^* n$ and every reduction from $t_2$ to $n$ has length
$k-1$, $p+q = k$, concluding the proof.
\end{proof}

An orthogonal TRS is said to be {\em variable 
preserving}\footnote{Some texts have a weaker notion
of {\em variable preserving}, but our version is more suitable for
investigating reduction lengths.}
if for every rule $\ell \to r$ every variable occurring in $\ell$
occurs exactly once in $r$.
For instance, the TRS consisting of the two rules for $\pl$ as
given in the introduction, is variable preserving.

\begin{lem}
\label{lembas}
Let $R$ be a variable preserving orthogonal TRS. Then 
the relation $\to_R$ on the set of all terms
satisfies the diamond property.
\end{lem}
\begin{proof}
Similar to the proof of the well-known critical pair lemma 
(see e.g., \cite{terese}, Lemma 2.7.15) we do a case analysis on the redex
patterns. If a term $t$ has one-step reductions to two distinct terms
$u$ and $v$, then by orthogonality one of the following cases
holds:
\begin{enumerate}
\item (the redexes are parallel) $t = C[\ell_1^\sigma, \ell_2^\tau]$
and $u = C[r_1^\sigma, \ell_2^\tau]$
and $v = C[\ell_1^\sigma, r_2^\tau]$, for a two-hole context $C$,
substitutions, $\sigma, \tau$ and rules $\ell_i \to r_i$ in $R$ for $i
= 1,2$,
\item (one redex is above the other) $t = C[\ell^\sigma]$ and
$u = C[r^\sigma]$ and $v = C[\ell^\tau]$, for a rule $\ell \to r$ in
$R$ and a context $C$, and substitutions $\sigma, \tau$ 
such that $x^\sigma \to_R x^\tau$ for some variable $x$ occurring 
in $\ell$, and $y^\sigma = y^\tau$ for all variables $y$ distinct from $x$,
\end{enumerate}
or the same in which $u$ and $v$ are swapped.
For both cases we have to find $w$ such that $u \to_R w$ and $v \to_R w$.
For case (1) this holds by choosing $w = C[r_1^\sigma, r_2^\tau]$.

For case (2) write $\ell = D[x]$. Since $R$ is left-linear, $x$ does
not occur in $D$. Since $R$ is variable preserving, we can write $r =
E[x]$ for a context $E$ not containing $x$. Now we obtain
\[u = C[r^\sigma] = C[ E^\sigma [x^\sigma]] \to_R C[ E^\sigma [x^\tau]]
\; \mbox{ and } \; 
v = C[\ell^\tau] \to_R C[ r^\tau] = C[E^\tau [x^\tau]].\]
Since $x$ does not occur in $E$ we have $E^\tau = E^\sigma$, by which
$w = C[E^\tau[x^\tau]] = C[E^\sigma [x^\tau]]$ satisfies the
requirements for $w$.
\end{proof}

As a direct consequence of Lemmas \ref{cordiam} and \ref{lembas} 
we conclude that
for all terms with respect to the two $\pl$ rules from
the introduction the number of steps to reach the normal form is
independent of the strategy. This does not hold any more for the
full system also containing the rules for $\ml$. For instance,
the term $\ml(\fz,\pl(\fz,\fz))$ admits the following two reductions
to normal form having lengths one and two, respectively: 
\[ \ml(\fz,\pl(\fz,\fz)) \to_R \fz, \;\; 
\ml(\fz,\pl(\fz,\fz)) \to_R \ml(\fz,\fz) \to_R \fz.\]

However, here the starting term $\ml(\fz,\pl(\fz,\fz))$ is not {\em
basic}: it contains more than one {\em defined symbol}.
A symbol is called a {\em defined symbol} if it occurs as the
root of the left-hand side of a rule. Similar as in texts on
run-time complexity like \cite{AM10} we
define:
\begin{quote}
A term is defined to be {\em basic} if the root is a defined symbol,
and it is the only defined symbol occurring in the term.
\end{quote}

We will prove that for basic terms like terms of the shape 
$\ml(\fs^m(\fz),\fs^n(\fz))$ every reduction to its normal form
$\fs^{mn}(\fz)$ has the same length. 

\begin{thm}
\label{thmmain}
Let $R$ be an orthogonal  TRS over $\Sigma$, and $\Sigma' \subseteq
\Sigma$, such that
every rule $\ell \to r$ of $R$ is of one of the following shapes:
\begin{itemize}
\item $\ell \to r$ is variable preserving, and neither $\ell$
nor $r$ contain symbols from $\Sigma'$,
\item the root of $\ell$ is in $\Sigma'$, and 
for every symbol in $r$ 
from $\Sigma'$ the arguments of this symbol do not contain
defined symbols.
\end{itemize}
Then any two reductions of a basic term to normal form have the
same length.
\end{thm}

\begin{proof}
First we prove the following claim.
\begin{quote}
{\bf Claim:}
If $t$ is a basic term and $t \to_R^* u$, then 
for every symbol from $\Sigma'$ in $u$ the arguments 
of this symbol do not contain defined symbols.
\end{quote}
We prove this claim by induction on the number of steps in the
reduction $t \to_R^* u$. If this number is 0, then $u = t$ and the
claim follows from the definition of basic term.

If this number is $> 0$, then $t \to_R^r u' \to_R u$, where by
induction hypothesis we assume that the claim holds for $u'$. 
For every subterm of $u$ having its root in $\Sigma'$, we have to
prove that this subterm does not contain other defined symbols.
For every such subterm that also occurs in $u'$ this holds by the
induction hypothesis. The only way such a subterm may occur in $u$ but
not in $u'$ is that it is created by replacing $\ell^\sigma$ by 
$r^\sigma$ for some substitution $\sigma$ and a rule $\ell
\to r$ of the second type: the root of $\ell$ is in $\Sigma'$ and
arguments of symbols from $\Sigma'$ in $r$ do not contain defined
symbols. Since $\ell^\sigma$ is a subterm of $u'$ and the property
holds for $u'$ by induction hypothesis, for all variables $x$ in $\ell$ 
the term $x^\sigma$ does not contain defined symbols.
Hence every subterm of $r^\sigma$ having its root in $\Sigma'$
does not contain other defined symbols, concluding the proof of the
claim.

Next we show that the relation $\to_R$ restricted to terms $u$
obtained by $t \to_R^* u$ for some basic term $t$,
satisfies the diamond property.
To prove this, let such a term $u$ both rewrite to $v$ and to
$w$, $v \neq w$. In case the redexes are parallel, then the
diamond property is easily concluded as in the proof of Lemma
\ref{lembas}, case 1.

In the remaining case one
redex is above the other. Due to the claim and the shape of
the rules, one of the following cases holds:
\begin{itemize}
\item The roots of both redexes are not in $\Sigma'$. Then both
reduction steps are with respect to variable preserving rules,
and the diamond property follows from Lemma \ref{lembas}.
\item The root of the innermost (of the two) redexes is in
$\Sigma'$, but the root of the outermost redex is not. Then the
reduction with respect to the outermost redex is variable
preserving, by which the diamond property can be concluded.
\end{itemize}
In all cases the diamond property of $\to_R$ restricted reducts $u$ 
of a basic term $t$ can be concluded. Now the
theorem follows from Lemma \ref{cordiam}.
\end{proof}

Indeed now by Theorem \ref{thmmain} we can conclude that 
with respect to the TRS with $\pl$ and $\ml$ as given
in the introduction for every basic term of the shape 
$\ml(\fs^m(\fz),\fs^n(\fz))$ every reduction to its normal form
$\fs^{mn}(\fz)$ has the same length: let $\Sigma'$ consist of the
single symbol $\ml$. Then the two rules for $\pl$ satisfy the
condition for rules of the first type, and the two rules for 
$\ml$ satisfy the condition for rules of the second type, by
which the claim follows from Theorem \ref{thmmain}.

As another example consider the Fibonacci function $\fib$
defined by the rules
\[ \begin{array}{rclrcl}
\pl(\fz,x) & \to & x &
\fib(\fz) & \to & \fz \\
\pl(\fs(x),y) & \to & \fs(\pl(x,y)) \hspace{1cm} &
\fib(\fs(\fz)) & \to & \fs(\fz) \\
&&& \fib(\fs(\fs(x))) & \to & \pl(\fib(x),\fib(\fs(x))). \end{array} \]
Choosing $\Sigma'$ to consist only of the symbol $\fib$, all
requirements of Theorem \ref{thmmain} hold, so we conclude
that for every $k$ every reduction of $\fib(\fs^k(\fz))$ to normal
form has the same length.

For another standard recursive function, the factorial, it does not 
hold that distinct reductions of basic terms to normal forms have the
same length. By extending our basic rules for $\pl$ and $\ml$ from the
introduction by the rules for factorial
\[ \begin{array}{rcl}
\ffact(\fz) & \to & \fs(\fz)  \\
\ffact(\fs(x)) & \to & \ml(\fs(x),\ffact(x)) \end{array} \]
the basic term $\ffact(\fs(\fz))$ admits reductions of lengths both 5
and 6 to its normal form $\fs(\fz))$, caused by the fact that after
\[ \ffact(\fs(\fz)) \to \ml(\fs(\fz),\ffact(\fz)) \to 
\pl(\ml(\fz,\ffact(\fz)),\ffact(\fz)) \]
the subterm $\ml(\fz,\ffact(\fz))$ can be rewritten to $\fz$ in one
step, but also in two steps via $\ml(\fz,\fs(\fz))$. Indeed 
Theorem \ref{thmmain} does not apply here since $\ml$ should be in
$\Sigma'$ as before, and also $\ffact \in \Sigma'$ since $\ml$ occurs in 
the right-hand side of a rule for $\ffact$. But then this rules
violates the conditions: $\ffact(x)$ occurs as an argument of the
symbol $\ml$ in this right-hand side.

The converse of Theorem \ref{thmmain} does not hold, not even for
terminating orthogonal exhaustive constructor systems. Here 
exhaustive (in case of termination equivalent to {\em sufficiently complete}) 
means that for every defined symbol applied to ground 
constructor forms at least one rule is applicable, by which all 
ground normal forms are constructor terms. 

For instance, consider the TRS consisting of the three rules 
\[ \ff(\fz) \to \fz, \;\;\; \ff(\fs(x)) \to \fs(\fz), \;\;\;
\fg(x) \to \ff(\ff(x))). \]
As $\ff$ and $\fg$ are defined symbols, the only basic terms are of
the shape $\ff(\fs^n(x))$, $\ff(\fs^n(\fz))$, 
$\fg(\fs^n(x))$ or $\fg(\fs^n(\fz))$, for $n \geq 0$. All of these
terms only admit one single reduction to normal form of at most three
steps, as is easily checked by applying the rules.
So any two reductions of a basic term to normal form have the
same length. However, this can not be concluded from Theorem
\ref{thmmain} since this TRS does not satisfy its conditions:
the rule $\ff(\fs(x)) \to \fs(\fz)$ is not
variable preserving, so for satisfying the conditions 
it should hold $\ff \in \Sigma'$. But then the third
rule does not satisfy the conditions as there is a nested occurrence of
$\ff$ in the right-hand side. 

\section{Binary arithmetic}
\label{secbin}
For more efficient computation of numbers it is natural to 
exploit binary notation, in which the size of the representation
is logarithmic rather than linear in the value of the number, and in
which basic arithmetic can be executed by rewriting in a complexity 
that is polynomial in the size of this representation, so logarithmic 
in the values of the numbers. Although the general idea is folklore,
there is not a single fixed standard. Here we present a
straightforward way to proceed for implementing binary arithmetic by
rewriting. It is related to the non-orthogonal system from \cite{WZ95} for
$n$-ary arithmetic for arbitrary $n$ for all integers. By restricting to 
$n=2$ and positive integers, here we succeed in presenting an orthogonal system.

In standard binary notation positive integers can be seen to be
uniquely composed from a constant $\fo$ representing value 1, and
two unary operators $\az$ and $\ao$, where $\az$
means putting a 0 behind the number, by which its value is
duplicated, and $\ao$
means putting a 1 behind the number, by which its value $x$ is
replaced by $2x+1$. 

Every positive integer has a unique representation as a ground
term over these three symbols $\fo$, $\az$ and $\ao$, corresponding 
to the usual binary notation
in which a postfix notation for $\az$ and $\ao$ is used. 
For instance, the number 
29 is $11101$ in binary notation, and is written as $\fo \ao \ao
\az \ao$ as a postfix
ground term. Since in term rewriting it is more standard to use prefix
notation rather than postfix notation, here we choose to switch to 
prefix notation. So
instead for 29 we write $\ao(\az(\ao(\ao(\fo))))$.
Introducing a constant $\fz$ would violate unicity, that is why we 
restrict to positive integers. 

In order to express addition and multiplication in this notation the
successor $\scp$ is needed as an additional operator, having 
rewrite rules
\[ \begin{array}{rcl}
\scp(\fo) & \to & \az(\fo) \\
\scp(\az(x)) & \to & \ao(x) \\
\scp(\ao(x)) & \to & \az(\scp(x)) \\
\end{array} \]
Now we can express addition:
\[ \begin{array}{rclrcl}
\plp(\fo,x) & \to & \scp(x) &
\plp(\az(x),\az(y)) & \to & \az(\plp(x,y)) \\
\plp(\az(x),\fo) & \to & \ao(x) &
\plp(\az(x),\ao(y)) & \to & \ao(\plp(x,y)) \\
\plp(\ao(x),\fo) & \to & \az(\scp(x))  \hspace{1cm} &
\plp(\ao(x),\az(y)) & \to & \ao(\plp(x,y)) \\
&&& \plp(\ao(x),\ao(y)) & \to & \az(\scp(\plp(x,y))) 
\end{array} \]
Indeed now for every ground term composed from $\fo, \az, \ao,
\scp, \plp$ containing at least one symbol $\scp$ or $\plp$ a
rule is applicable. So the ground normal forms are the ground
terms composed from $\fo, \az, \ao$, exactly being the binary
representations of positive integers. 
Since the TRS is easily proved to be
terminating, e.g., by recursive path order, every ground term
will reduce to such a ground normal form, being the binary
representations of a positive integer. As by every rule the
numeric value of the term is preserved, this TRS serves for
computing the binary value of any ground term considered so far.

Surprisingly, it is very simple to extend this system to
multiplication, by using the fresh symbol $\mlp$ for
multiplication and introducing the following rules
\[ \begin{array}{rcl}
\mlp(\fo,x) & \to & x \\
\mlp(\az(x),y) & \to & \az(\mlp(x,y)) \\
\mlp(\ao(x),y) & \to & \plp(\az(\mlp(x,y)),y) \\
\end{array} \]

The above observations are also easily checked for the extended
TRS consisting of all rules presented so far: all rules preserve
values, and for
every ground term composed from $\fo, \az, \ao, \scp, \plp, \mlp$ 
containing at least one symbol $\scp$ or $\plp$ or $\mlp$ a
rule is applicable. So the ground normal forms are the ground
terms composed from $\fo, \az, \ao$, exactly being the binary
representations of positive integers. Also the extended TRS 
is easily proved to be terminating, e.g., by recursive path order, 
so every ground term
will reduce to such a ground normal form, being the binary
representations of a positive integer. So the extended TRS
also serves for executing multiplication.

Choosing $\Sigma'$ to consist of the single symbol $\mlp$ it is
easily checked that all conditions of Theorem \ref{thmmain} are
satisfied. So by Theorem \ref{thmmain} we conclude that for
using this TRS for computing the addition or multiplication of
two binary numbers, every reduction to normal form has the same
length. This observation is of great value: a lot of effort is
done in choosing the right rewriting strategy; this observation shows
that when binary arithmetic is only used for addition and
multiplication (which is often the case), strategy optimization is 
useless since all reductions to normal form are of the same complexity
as they all have the same length.

In \cite{Zarit} it is proved that for addition this reduction
length is linear in the size of the arguments, and for
multiplication it is quadratic in the size of the arguments, so
having the same complexity as the standard algorithms for binary
addition and multiplication. The TRSs as presented in
\cite{Zarit} have been the basis of the implementation of integer
number computation in the mCRL2 tool set \cite{mcrl2}.

In our rules for multiplication the recursion is only in the first
argument. For computing $2^n * m$ for any number $m$ this is very 
efficient: the result is obtained after $O(n)$ steps, independent of
the value of $m$. However, for computing $m * 2^n$ this is not
efficient at all: first $m$ will be decomposed and many copies of
$2^n$ will be created, finally giving the same result as the
computation of $2^n * m$. To bring more symmetry and to make the
computations of $m * 2^n$ and $2^n * m$ both be linear in $n$, we can
choose the rules for $\mlp$ to be
\[ \begin{array}{rcl}
\mlp(\fo,\fo) & \to & \fo \\
\mlp(\az(x),\fo) & \to & \az(x)) \\
\mlp(\ao(x),\fo) & \to & \ao(x)) \\
\mlp(\fo,\az(x)) & \to & \az(x)) \\
\mlp(\fo,\ao(x)) & \to & \ao(x)) \\
\mlp(\az(x),\az(y)) & \to & \az(\az(\mlp(x,y))) \\
\mlp(\az(x),\ao(y)) & \to & \az(\mlp(x,\ao(y))) \\
\mlp(\ao(x),\az(y)) & \to & \az(\mlp(\ao(x),y)) \\
\mlp(\ao(x),\ao(y)) & \to & \plp(\az(\mlp(x,\ao(y))),\ao(y)). 
\end{array} \]
In this way duplication of arguments and introducing $\plp$ is only
done if both arguments of $\mlp$ are of the shape $\ao(\cdots)$.
Although slightly more complicated, it is easy to check that
the full system consisting of the rules for $\scp$, $\plp$ and these
new rules for $\mlp$ 
\begin{itemize}
\item only rewrites terms to terms having the same value,
\item is terminating (again by recursive path order),
\item satisfies the conditions of Theorem \ref{thmmain} for $\Sigma' =
\{\mlp\}$, and
\item computes both $m * 2^n$ and $2^n * m$ in $O(n)$ rewrite steps,
independent of the value of $m$.
\end{itemize}
Hence by Theorem \ref{thmmain}, we conclude that 
for this improved system for binary arithmetic for every basic term
any two reductions to normal form have the same length.

\section{Conclusion}
\label{secconcl}
Basic terms are typical terms to be rewritten: a defined symbol
on top and constructor terms as arguments. We gave a criterion
for orthogonal TRSs by which all reductions of a basic term to
normal form have the same length, showing that the reduction
length is independent of the chosen strategy. This applies to
both addition and multiplication, both in unary and binary
notation. This result is surprising: it shows that for these basic
computations as they often occur in practice, all possible reduction
strategies have the same complexity.

In unary notation the Fibonacci function still satisfies our
criteria. However, for more complicated user defined functions
experiments based on a simple implementation show that 
reduction lengths of basic terms
are typically not strategy independent any more.

{\bf Acknowledgment}. We want to thank Evans Kaijage for
fruitful discussions on this topic and for doing some
experiments.

\bibliographystyle{eptcs}
\bibliography{ref}
\end{document}